\RequirePackage{fix-cm}
\documentclass{svjour3}                     
\smartqed  
\usepackage{graphicx}
\usepackage{cite}
\usepackage{amsmath,amssymb,amsfonts}
\usepackage[misc]{ifsym}
\usepackage{algorithmic}
\usepackage{textcomp}
\usepackage{algorithm}
\usepackage{multirow}
\usepackage{xcolor}
\usepackage{bm}
\def\BibTeX{{\rm B\kern-.05em{\sc i\kern-.025em b}\kern-.08em
    T\kern-.1667em\lower.7ex\hbox{E}\kern-.125emX}}
\usepackage{todonotes}
\usepackage{textcomp}
\usepackage{color}
\usepackage{bbm}
\usepackage{stfloats}
\usepackage{float}
\usepackage{subfigure}
\usepackage{caption}
\usepackage{pdfpages}
\newcommand{\bdr}{\bm{r}}
\newcommand{\bdw}{\bm{w}}

\newcommand{\alg}{IAS-II } 
\begin{document}

\title{A Convergent Primal-Dual Algorithm for Computing Rate-Distortion-Perception Functions
}


\author{Chunhui Chen\textsuperscript{*} \and
  Linyi Chen\textsuperscript{*} \and
        Xueyan Niu \and
        Hao Wu \textsuperscript{\Letter} 
}

\institute{
* These authors contributed equally to this work.\\[1ex]
Chunhui Chen  \at
              Department of Mathematical Sciences, Tsinghua University, Beijing, China 
           \and
Linyi Chen \at
              Department of Mathematical Sciences, Tsinghua University, Beijing, China
            \and
Xueyan Niu \at
            Theory Lab, 2012 Labs, Huawei Technologies Co., Ltd., Beijing, China
            \and
\Letter $\;$   Hao Wu \at
    Department of Mathematical Sciences, Tsinghua University, Beijing, China \\
    \email{hwu@mail.tsinghua.edu.cn}
}

\date{Received: date / Accepted: date}

\maketitle

\begin{abstract}
Recent advances in Rate-Distortion-Perception (RDP) theory highlight the importance of balancing compression level, reconstruction quality, and perceptual fidelity. While previous work has explored numerical approaches to approximate the information RDP function, the lack of theoretical guarantees remains a major limitation, especially in the presence of complex perceptual constraints that introduce non-convexity and computational intractability. Inspired by our previous constrained Blahut–Arimoto algorithm for solving the rate-distortion function, in this paper, we present a new theoretical framework for computing the information RDP function by relaxing the constraint on the reconstruction distribution and replacing it with an alternative optimization approach over the reconstruction distribution itself. This reformulation significantly simplifies the optimization and enables a rigorous proof of convergence. Based on this formulation, we develop a novel primal-dual algorithm with provable convergence guarantees. Our analysis establishes, for the first time, a rigorous convergence rate of \( O(1/n) \) for the computation of RDP functions. The proposed method not only bridges a key theoretical gap in the existing literature but also achieves competitive empirical performance in representative settings. These results lay the groundwork for more reliable and interpretable optimization in RDP-constrained compression systems. Experimental results demonstrate the efficiency and accuracy of the proposed algorithm.
\end{abstract}

\section{Introduction}

The interplay among compression level, reconstruction fidelity, and perceptual accuracy has emerged as a pivotal focus in contemporary lossy compression frameworks \cite{blau2018perception,balle2017endtoend,theis2017lossy,blau2019rethinking,zhang2021universal,salehkalaibar2023choice,niu2025rate,chen2025information}. Classical Rate-Distortion (RD) theory, rooted in Shannon’s foundational work \cite{shannon1959coding}, provides a rigorous framework to study the fundamental limits of lossy data compression \cite{DBLP:books/wi/01/CT2001}. However, traditional formulations assume distortion metrics (e.g., mean squared error) that do not align well with human perceptual judgments \cite{wang2004image,blau2018perception}. This mismatch has become increasingly pronounced in high-dimensional sensory data, such as images and video, where perceptual quality often plays a more critical role than pixel-wise accuracy.

Recent advances in machine learning, particularly the emergence of deep generative models (e.g., GANs \cite{goodfellow2014generative}, VAEs \cite{kingma2013auto}, diffusion models \cite{sohl2015deep,song2019generative}), have reinvigorated interest in perceptual-aware compression techniques. These models enable the synthesis of visually plausible reconstructions even at low bit-rates \cite{10827256,jia2024generative}, motivating a shift from fidelity-centric to perception-centric evaluation metrics. In this context, \emph{Rate-Distortion-Perception} (RDP) theory has emerged as a unifying framework to quantify the three-way trade-off between rate, distortion measure, and perceptual divergence \cite{blau2019rethinking,agustsson2019generative,niu2023conditional,zhang2021universal,10619317,chen2025information}. Within this framework, the information RDP functions generalize the classical RD functions by introducing an additional constraint on the perceptual distance -- typically measured using total variation, KL divergence, or Wasserstein distance \cite{zhang2021universal,chen2023computation} -- between the distribution of reconstructions and that of the original source.

Despite the theoretical appeal of RDP, computing the information RDP function remains challenging in practice. This difficulty stems primarily from the perceptual constraint, which induces the non-strict convexity and nonlinearity which break the inherent simplex structure of the RD functions, making it intractable for Blahut–Arimoto (BA) type algorithm. 
Although some recent efforts have adopted adversarial training strategies \cite{blau2019rethinking,zhang2021universal}, these methods usually rely on heuristic procedures and lack rigorous convergence guarantees. Furthermore, perception constraints based on optimal transport distances impose additional computational burdens, as they require the solution of large-scale transport problems within the inner loop of RDP optimization \cite{villani2009wasserstein,cuturi2013sinkhorn}.

Since the communication process can essentially be viewed as a transportation problem, in recent years there has been a growing interest in incorporating optimal transport techniques into information theory. This line of research has achieved remarkable success including the analysis of LM rate \cite{10000926}, RD functions \cite{wu2023communication}, and the information bottleneck framework \cite{10206826}. Motivated by these developments, we extended the optimal transport perspective to the study of RDP functions.
In our previous work, we proposed the Wasserstein Barycenter Model for RDP functions (WBM-RDP), and introduced the Improved Alternating Sinkhorn (IAS) algorithm, which effectively computes RDP functions under different perception measures \cite{chen2023computation}. Although the IAS algorithm exhibited strong empirical performance, its convergence was established only by numerical evidence, without rigorous theoretical guarantees.

In this paper, we address this critical limitation by introducing a novel theoretical framework for computation of the RDP function. Inspired by recent progress in Constrained Blahut–Arimoto (CBA) algorithm \cite{DBLP:journals/tit/Arimoto72,blahut1972computation, chen2023constrained}, we revisit the formulation of the perceptual constraint and propose a key relaxation strategy. Instead of directly imposing hard constraints on the reconstruction distribution, we reformulate the problem by adopting an alternative optimization framework over the reconstruction distribution space. This relaxation preserves the optimal solution to the original problem while offering substantially improved analytical tractability. In particular, when the perception measure is specified as the Wasserstein metric, the reformulation transforms the original Wasserstein Barycenter (WB) model into a relaxed WB formulation.

Using this reformulation, we design a primal-dual optimization algorithm tailored to the relaxed WB structure. Our method enables efficient computation with provable convergence guarantees. To our knowledge, this is the first work that rigorously establishes a convergence rate of $O(1/n)$ for RDP optimization under different perception measures, i.e., Wasserstein metric, total variation (TV) distance, and Kulback-Leibler (KL) divergence. This result closes a long-standing theoretical gap in the literature, where prior methods have largely relied on heuristic numerical schemes without guarantees. In addition, our algorithm exhibits strong empirical performance: numerical experiments demonstrate that it outperforms prior approaches in both computational efficiency and solution quality, while also providing reliable theoretical guarantees.

The remainder of this paper is organized as follows. 
In Section~\ref{sec:RDP}, we formally define the information RDP function and present its discrete formulation under general distortion and perception measures. 
Section~\ref{sec:algorithm} introduces the proposed primal--dual algorithms, with separate designs for the KL divergence and Wasserstein metric cases, leveraging the structural properties of each measure. 
In Section~\ref{sec:convergence}, we provide a rigorous convergence analysis of the algorithms, establishing an \(O(1/n)\) convergence rate for both cases. 
Section~\ref{sec:experiment} reports numerical experiments on binary and Gaussian sources, demonstrating the computational efficiency and precision of the proposed method compared to existing approaches. 
Finally, Section~\ref{sec:conclusion} concludes the paper and discusses potential directions for future research.

\section{The RDP functions} \label{sec:RDP}
Consider a discrete memoryless source $X \in \mathcal{X}$ and a reconstruction $\hat{X} \in \hat{\mathcal{X}}$, where $\mathcal{X}=\{x_1,\cdots, x_M\},\hat{\mathcal{X}}=\{\hat{x}_1,\cdots, \hat{x}_N\}$ are finite alphabets. 
Suppose $p_X$ and $p_{\hat{X}}$ are defined in the probability space $(\mathcal{X}, \mathcal{F}, \mathbb{P}).$ We consider the single-letter distortion measure
$\Delta: \mathcal{X}\times \hat{\mathcal{X}}\mapsto [0,\infty)$ and the perception measure between the distributions
$d:\mathbb{P}\times \mathbb{P}\mapsto [0,\infty).$

\begin{definition}[The information RDP function]\label{definition1}
Given a distortion fidelity $D$ and a perception quality $P$, the information RDP function is defined as\par 
\begin{subequations}\label{eq0}
\begin{align}
R(D, P)= \min _{p_{\hat{X} \mid X}} \quad& I(X, \hat{X}) \label{eq0_a} \vspace{1ex} \\
 \text { s.t. }\quad &\mathbb{E}[\Delta(X, \hat{X})] \leq D,\label{eq0_b} \vspace{1ex} \\  
 &  d\left(p_X, p_{\hat{X}}\right) \leq P, \label{eq0_c}
\end{align}
\end{subequations}%
where the minimization is taken over all conditional distributions, and $I(X,\hat{X})$ is the mutual information. 
\end{definition}
Note that the information RDP function \eqref{eq0} degenerates to RD functions when the constraint \eqref{eq0_c} is removed.

Since the alphabets $\mathcal{X}$ and $\hat{\mathcal{X}}$ are finite, we denote 
\[p_i = p_X(x_i),\ r_j = p_{\hat{X}}(\hat{x}_j),\  d_{ij} = \Delta(x_i,\hat{x}_j)\]
and $w_{ij}=W(\hat{x}_j \mid x_i)$ for all $1\le i\le M, 1\le j\le N$. Here, $W: \mathcal{X} \rightarrow \hat{\mathcal{X}}$ is the channel transition mapping.
Thus the discrete form of problem \eqref{eq0} can be written as
\begin{subequations} \label{eq1}
\begin{align}
    \min _{\bm{w},\bm{r}} \quad  &\sum_{i=1}^M \sum_{j=1}^N\left(w_{i j} p_i\right)\left[\log w_{i j}-\log r_j\right] \vspace{1ex} \label{eq1_a}\\
    \text { s.t. }  &\sum_{j=1}^N w_{i j}=1,\  \sum_{j=1}^N r_j=1 , \  \forall i\vspace{1ex} \label{eq1_b}\\
    &\sum_{i=1}^M \sum_{j=1}^N w_{i j} p_i d_{i j} \leq D,\   \vspace{1ex} \label{eq1_c}\\
    &d(\bm{p},\bm{r}) \leq P. \label{eq1_d}
\end{align}
\end{subequations}
Unlike WBM-RDP in \cite{chen2023computation}, here we follow the idea of CBA model in \cite{chen2023constrained} to eliminate the constraint of the reconstruction distribution by adopting an alternative optimization strategy along the $r_j$ direction, as elaborated in the following subsections.

One commonly used measure of perceptual quality $d(\bm{p},\bm{r})$ is the Wasserstein metric,
\begin{subequations} \label{metric}
\begin{align}
\mathcal{W} (\bm{p},\bm{r}) =  &\min_{\bm{\Pi}}\quad\sum_{i=1}^M \sum_{j=1}^N \Pi_{i j} c_{i j} \vspace{1ex}\\ 
      &\text { s.t. } \sum_{i=1}^M \Pi_{i j}= r_j ,\  \sum_{j=1}^N \Pi_{i j}= p_i,\  \forall i, j,
\end{align}
\end{subequations}%
where $c_{ij}$ denotes the cost matrix between $x_i$ and $\hat{x}_j$.  In some cases, the TV distance $\delta(\bm{p}, \bm{r}) = \frac{1}{2}\|\bm{p}-\bm{r} \Vert_1$ and the Kulback-Leibler (KL) divergence $\text{KL}(\bm{p}\| \bm{r}) = \sum_{i=1}^M p_i \left[\log p_{i}-\log r_i\right]$ are also considered as a measure of perceptual quality.

\begin{remark}
    If the perception measure is substituted by the TV distance, we only need to set the cost matrix as $c_{ij} = \mathbf{1}_{i \neq j}$ (see Eq. (6.11) of \cite{villani2009optimal}). Thus, TV distance can be seen as a special case of the Wasserstein metric.  
\end{remark}

\section{The \alg algorithm} \label{sec:algorithm}

In this section, we present the newly proposed primal–dual algorithms for the aforementioned model \eqref{eq1}, which are developed as an extension of the original IAS algorithm and therefore are called the \emph{IAS-II algorithm}.
Since the KL divergence admits a closed-form expression, whereas the Wasserstein metric is inherently defined as the solution to the problem \eqref{metric}, we design separate algorithms tailored to the specific characteristics of each measure in the following subsections.

\subsection{KL divergence case}
In this case, the problem \eqref{eq1} can be written as 
\begin{subequations} \label{eq2}
\begin{align}
    \min _{\bm{w},\bm{r}} \quad  &\sum_{i=1}^M \sum_{j=1}^N\left(w_{i j} p_i\right)\left[\log w_{i j}-\log r_j\right] \vspace{1ex} \label{eq2_a}\\
    \text { s.t. }  &\sum_{j=1}^N w_{i j}=1,\  \sum_{j=1}^N r_j=1 , \  \forall i\vspace{1ex} \label{eq2_b}\\
    &\sum_{i=1}^M \sum_{j=1}^N w_{i j} p_i d_{i j} \leq D,\   \vspace{1ex} \label{eq2_c}\\
    &\sum_{j=1}^M p_j \log\left(\sum_{i=1}^Mp_iw_{i j}\right) \geq T, \label{eq2_d}
\end{align}
\end{subequations}
where $T = \sum_{i=1}^M p_i \log p_{i} - P$. Thus, the above optimization problem can be seen as a double-minimization problem.

First, fix the variable $\bm{w}$, and minimize \eqref{eq2} as an optimization problem with respect to the variable $\bdr$ only.
By introducing the multiplier $\eta$, the Lagrangian function can be written as: 
\begin{equation*}
    \mathcal{L}(\bm{r},\eta)=\sum_{i=1}^M \sum_{j=1}^N\left(w_{i j} p_i\right)\left[\log w_{i j}-\log r_j\right]+\eta(\sum_{j=1}^N r_j-1).
\end{equation*}
The first-order condition gives
\begin{equation*}
\frac{\partial \mathcal{L}}{\partial r_{j}}=-\sum_{i=1}^{M} w_{i j}p_{i} \frac{1}{r_{j}}+\eta=0,
\end{equation*}
and it further yields the representation of $\bdr$ by the multiplier $\eta$:
\begin{equation}\label{from_of_r}
r_{j} = \left(\sum_{i=1}^{M} w_{i j}p_{i} \right)\Big/\eta.
\end{equation}
Substituting \eqref{from_of_r} into the equality constraint $\sum_j r_j=1$, we have that the multiplier $\eta$ should be updated to satisfy the following equation:
\begin{equation*}
F(\eta) \triangleq \sum_{j=1}^{N}\left[\left(\sum_{i=1}^{M} w_{i j}p_{i} \right)\Big/\eta\right]-1=0. \label{F_def}
\end{equation*}
Here, $F(\eta)$ is a monotonic function with a unique real root $\eta=1$, due to the fact that 
\begin{equation*}
\sum_{i=1}^{M}\sum_{j=1}^{N} w_{ij}p_{i}=1.
\end{equation*}
Hence, we obtain the optimal $\bdr$ under fixed $\bdw$, as
\begin{equation*}
{r}_j=\sum_{i=1}^{M} w_{i j}p_{i}, \quad j = 1, \cdots, N.
\end{equation*}
Second, fix the variable $\bdr$, and minimize \eqref{eq2} as an optimization problem with respect to the variable $\bdw$ only. We have to solve the following optimization sub-problem:
\begin{subequations} \label{eq3}
\begin{align}
    \min _{\bm{w}} \quad  &\sum_{i=1}^M \sum_{j=1}^N\left(w_{i j} p_i\right)\left[\log w_{i j}-\log r_j\right] \vspace{1ex} \label{eq3_a}\\
    \text { s.t. }  &\sum_{j=1}^N w_{i j}=1,\  \  \forall i\vspace{1ex} \label{eq3_b}\\
    &\sum_{i=1}^M \sum_{j=1}^N w_{i j} p_i d_{i j} \leq D,\   \vspace{1ex} \label{eq3_c}\\
    &\sum_{j=1}^M p_j \log\left(\sum_{i=1}^Mp_iw_{i j}\right) \geq T, \label{eq3_d}
\end{align}
\end{subequations}
By introducing the multipliers $\bm{b}, \lambda\geq 0,\gamma\geq 0$, the Lagrange function of \eqref{eq3} can be written as:
\begin{multline}\label{Lagrangian}
    \mathcal{L} \left(\bm{w};\bm{b}, \lambda,\gamma \right) =
    \sum_{i=1}^M \sum_{j=1}^N w_{i j} p_i \log \frac{w_{i j}}{r_j} 
    + \sum_{i=1}^M b_i \left(\sum_{j=1}^N w_{ij} - 1\right) \\
    + \lambda \left(\sum_{i=1}^M \sum_{j=1}^N w_{i j} p_i d_{i j} - D\right)
    - \gamma \left( p_j \log\left(\sum_{i=1}^M p_i w_{i j}\right) - T \right)
\end{multline}
Taking the derivative of $\mathcal{L} \left(\bm{w};\bm{b}, \lambda,\gamma \right)$ with respect to the primal variable $\bm{w}$, one obtains:
\begin{equation*} 
\begin{aligned}
    \frac{\partial \mathcal{L}}{\partial w_{i j}}&=p_i(1+\log w_{i j}-\log r_j)+\lambda p_i d_{ij}-\gamma p_j p_i/\left(\sum_{i=1}^M p_iw_{i j}\right)+b_i,
\end{aligned}
\end{equation*}
and further yields the representation of $\bdw$ by dual variables
\begin{equation}\label{w}
w_{ij} = r_je^{a_j-\lambda d_{ij}-b_i/p_i-1},
\end{equation}
where
\begin{equation}\label{a}
a_{j} = \frac{\gamma p_j}{\sum_{i=1}^Mp_iw_{ij}},
\end{equation}
combining \eqref{w} and \eqref{a}, we can obtain
\begin{equation}\label{ac}
a_{j} = \frac{\gamma p_j}{\sum_{i=1}^Mp_ir_je^{a_j-\lambda d_{ij}-b_i/p_i-1}},
\end{equation}
Then we can obtain $a_j$ by finding the root of 
\begin{equation*}
    f_j(a_j) =  p_j e^{-a_j} - 1/\gamma\sum_{i=1}^M p_ir_je^{-\lambda d_{ij}-b_i/p_i-1} a_{j}
\end{equation*}
by Newton's method in $\mathbb{R}^{+}$.
Substituting \eqref{w} into \eqref{eq3_b}, we can update $b_i$ as follows:
\begin{equation}\label{b}
    b_i=p_i\log(\sum_j r_j e^{a_j-\lambda d_{ij}-1}).
\end{equation}

To update the dual variables $\lambda$ and $\gamma$, we can use a similar technique in \cite{chen2023constrained}, substituting \eqref{a} into \eqref{eq3_d} and \eqref{w} into \eqref{eq3_c} respectively, and we obtain the inner iterations:
\begin{itemize}
    \item Find the root of  $$f_j(a_j)  =  p_j e^{-a_j} -  1/\gamma \sum_{i=1}^M p_ir_j \allowbreak e^{-\lambda d_{ij}-b_i/p_i-1} a_{j}$$ by Newton's method in $\mathbb{R}^{+}$,
    \item Update $\bm{b}$ by $$b_i=p_i\log(\sum_j r_j e^{a_j-\lambda d_{ij}-1}).$$
    \item Update $\gamma$ by $$\gamma = \exp({\sum_{i=1}^M p_i \log a_{i}-P}).$$
    \item Find the root of $$g(\lambda) = \sum_{i=1}^M p_i\sum_{j=1}^N  r_jd_{ij} \allowbreak e^{a_j-\lambda d_{ij}-b_i/p_i-1}- D$$ by Newton's method in $\mathbb{R}^{+}$.
\end{itemize}

After the alternating iterations, the optimization problem \eqref{eq3} will be solved, and we will get the optimal solution $\bm{w}$. 

The \alg algorithm for the KL case can be summarized in Alg \ref{alg:KL}.
\begin{algorithm}[ht]
	\caption{IAS-II Algorithm for KL case}
	\label{alg:KL}
	\begin{algorithmic}[1]
		\REQUIRE Distribution $p_{i}$, distortion measure $d_{ij}$,
        maximum iteration number $max\_iter$.
		\ENSURE Minimal value $\sum_{i=1}^{M} \sum_{j=1}^{N} w_{i j}p_i(\log w_{i j}-\log r_j)$.
        \STATE \textbf{Initialization:} $\bm{r}=\frac{1}{N}\mathbf{1}_{N},\bm{a}=\mathbf{1}_{N},\bm{b}=\mathbf{1}_{M},\gamma,\lambda=1$
		\FOR{$\ell = 1 : max\_iter$}
		\STATE $r_{j} \gets \sum_i p_i w_{ij}$
		\WHILE{the inner iteration is sufficient}
            \STATE Find root of $f_j(a_j)$ by Newton's method in $\mathbb{R}^{+}$
		\STATE $b_i\gets p_i\log(\sum_j r_j e^{a_j-\lambda d_{ij}-1})$
		\STATE $\gamma \gets \exp({\sum_{i=1}^M p_i \log a_{i}-P})$
            \STATE Find root of $g(\lambda)$ by Newton's method in $\mathbb{R}^{+}$
            \ENDWHILE
            \STATE $w_{ij} \gets r_je^{a_j-\lambda d_{ij}-b_i/p_i-1}$
		\ENDFOR
		\STATE \textbf{end}
		\RETURN $\sum_{i=1}^{M} \sum_{j=1}^{N} w_{i j}p_i(\log w_{i j}-\log r_j)$
	\end{algorithmic}
\end{algorithm}\\

\subsection{Wasserstein distance case}
Unlike the KL divergence, the Wasserstein metric does not have a closed-form expression; instead, it is itself the solution to ~\eqref{metric}. To address this, we introduce auxiliary optimization variables to enable an equivalent transformation, thereby reformulating the nested optimization problem~\eqref{eq1} into a unified simplex structure. The detailed formulation is provided in the following theorem.

\begin{theorem}
\label{thm-rdp}
The optimal value of \eqref{eq1} is equal to the optimal value of the following optimization problem \eqref{wasserstein}. Meanwhile, the optimal solution $(\bm{w},\bm{r})$ of the following optimization problem \eqref{wasserstein_0} is the optimal solution of \eqref{eq1}:
\begin{subequations}\label{wasserstein_0}
\begin{align}
    \min _{\bm{w},\bm{r},\bm{\Pi}} \quad &\sum_{i=1}^M \sum_{j=1}^N\left(w_{i j} p_i\right)\left[\log w_{i j}-\log r_j\right]     \label{wasserstein_0_a}\vspace{1ex}\\
    \text { s.t. }  &\sum_{j=1}^N w_{i j}=1,\   \sum_{i=1}^M p_iw_{i j} =\sum_{i=1}^M \Pi_{i j} , \label{wasserstein_0_b}\vspace{1ex}\\
    & \sum_{j=1}^N \Pi_{i j}= p_i, \ \sum_{j=1}^N r_j=1. \label{wasserstein_0_c}\vspace{1ex}\\
    &\sum_{i=1}^M \sum_{j=1}^N w_{i j} p_i d_{i j} \leq D,\label{wasserstein_0_d}\vspace{1ex}\\
    &\sum_{i=1}^M \sum_{j=1}^N \Pi_{i j} c_{i j} \label{wasserstein_0_e}\leq P,
\end{align}
\end{subequations}
\end{theorem}

\begin{proof}
    Suppose that the optimal solution to (\ref{wasserstein_0}) is $\{ \bm{w}^*,\bm{r}^*,\bm{\Pi}^* \}$. 
    First, according to the property of the Wasserstein metric and the constraint, we can get 
    \begin{equation*}
        \mathcal{W}(\bm{p},\bm{r^*}) \leq \sum_{i=1}^M \sum_{j=1}^N \Pi^*_{i j} c_{i j} \leq D,
    \end{equation*}    
    thus $\{ \bm{w}^*,\bm{r}^* \}$ are the feasible solution to (\ref{eq1}). 
    
    Let $\{ \hat{\bm{w}},\hat{\bm{r}} \}$ be the the optimal solution to (\ref{eq1}), and $\hat{\bm{\Pi}}$ be the optimal transport plan of the Wasserstein metric between $\{ \hat{\bm{w}},\hat{\bm{r}} \}$, we denote 
     \begin{equation*}
    L(\bm{w},\bm{r})= \sum_{i=1}^M \sum_{j=1}^N\left(w_{i j} p_i\right)\left[\log w_{i j}-\log r_j\right] .
    \end{equation*}
    If $L(\hat{\bm{w}},\hat{\bm{r}}) < L(\bm{w}^*,\bm{r}^*)$, then $\{ \hat{\bm{w}},\hat{\bm{r}}, \hat{\bm{\Pi}} \}$ is the feasible solution to (\ref{wasserstein_0}) whose target value is lower than that of $\{ \bm{w}^*,\bm{r}^*,\bm{\Pi}^*\}$, which leads
    to contradiction. 
    
    Therefore $\{ \bm{w}^*,\bm{r}^* \}$ is the optimal solution to (\ref{eq1}).
\end{proof}

We notice that the above problem \eqref{wasserstein_0} is not strictly convex on $\bm{\Pi}$. Although the optimal value of \eqref{wasserstein_0} is unique, the corresponding optimal solutions can vary in the dimensions of $\bm{\Pi}$. This variability can affect both the convergence behavior and the computational speed of numerical solvers. Therefore, we incorporate the entropy regularization technique from optimal transport by introducing an entropy regularization coefficient $\varepsilon$ together with an entropy term:
$$H(\bm{\Pi}) = \sum_{i=1}^M \sum_{j=1}^N \Pi_{i j} \log(\Pi_{i j}).$$

Following the convergence proof for entropic optimal transport in~\cite{nutz2022entropic}, it can be shown that as $\varepsilon \to 0$, the solution to the entropy-regularized problem converges to that of the original problem. Moreover, entropy regularization renders the problem strictly convex, which in turn facilitates the design of efficient algorithms. The final entropy-regularized problem can be written as:
\begin{subequations}\label{wasserstein}
\begin{align}
    \min _{\bm{w},\bm{r},\bm{\Pi}} \quad &\sum_{i=1}^M \sum_{j=1}^N\left(w_{i j} p_i\right)\left[\log w_{i j}-\log r_j\right] + \varepsilon H(\bm{\Pi})    \label{wasserstein_a}\vspace{1ex}\\
    \text { s.t. }  &\sum_{j=1}^N w_{i j}=1,\   \sum_{i=1}^M p_iw_{i j} =\sum_{i=1}^M \Pi_{i j} , \label{wasserstein_b}\vspace{1ex}\\
    & \sum_{j=1}^N \Pi_{i j}= p_i, \ \sum_{j=1}^N r_j=1. \label{wasserstein_c}\vspace{1ex}\\
    &\sum_{i=1}^M \sum_{j=1}^N w_{i j} p_i d_{i j} \leq D,\label{wasserstein_d}\vspace{1ex}\\
    &\sum_{i=1}^M \sum_{j=1}^N \Pi_{i j} c_{i j} \label{wasserstein_e}\leq P,
\end{align}
\end{subequations}

The above optimization problem can also be seen as a double-minimization problem with respect to $(\bm{w},\bm{\Pi})$ and $\bm{r}$. Unlike the WBM-RDP in \cite{chen2023computation}, we only have strong constraints on one side for both couplings $\text{diag}(\bm{p})\cdot\bm{w}, \bm{\Pi}$ and do not directly require both to be equal to $\bm{r}$ for the other side, so the problem can also be seen as a relaxation of the Wasserstein-Barycenter problem.

First, fix the variable $\bm{w}$, and minimize \eqref{wasserstein} as an optimization problem with respect to the variable $\bdr$ only.
By introducing the multiplier $\eta$, the Lagrangian function can be written as:
\begin{equation*}
    \mathcal{L}(\bm{r},\eta)=\sum_{i=1}^M \sum_{j=1}^N\left(w_{i j} p_i\right)\left[\log w_{i j}-\log r_j\right]+\eta(\sum_{j=1}^N r_j-1).
\end{equation*}
The first-order condition gives
\begin{equation*}
\frac{\partial \mathcal{L}}{\partial r_{j}}=-\sum_{i=1}^{M} w_{i j}p_{i} \frac{1}{r_{j}}+\eta=0,
\end{equation*}
and it further yields the representation of $\bdr$ by the multiplier $\eta$:
\begin{equation}\label{from_of_r1}
r_{j} = \left(\sum_{i=1}^{M} w_{i j}p_{i} \right)\Big/\eta.
\end{equation}
Substituting \eqref{from_of_r1} into the equality constraint $\sum_j r_j=1$, we have that the multiplier $\eta$ should be updated to satisfy the following equation:
\begin{equation*}
F(\eta) \triangleq \sum_{j=1}^{N}\left[\left(\sum_{i=1}^{M} w_{i j}p_{i} \right)\Big/\eta\right]-1=0. \label{F_def}
\end{equation*}
Here, $F(\eta)$ is a monotonic function with a unique real root $\eta=1$, due to the fact that 
\begin{equation*}
\sum_{i=1}^{M}\sum_{j=1}^{N} w_{ij}p_{i}=1.
\end{equation*}
Hence, we obtain the optimal $\bdr$ under fixed $\bdw$, as
\begin{equation*}
{r}_j=\sum_{i=1}^{M} w_{i j}p_{i}, \quad j = 1, \cdots, N.
\end{equation*}

Second, fix the variable $\bdr$, and minimize \eqref{wasserstein} as an optimization problem with respect to the variable $\bdw,\bm{\Pi}$ only. We have to solve the following optimization sub-problem:
\begin{subequations}\label{wasserstein1}
\begin{align}
    \min _{\bm{w},\bm{\Pi}} \quad &\sum_{i=1}^M \sum_{j=1}^N\left(w_{i j} p_i\right)\left[\log w_{i j}-\log r_j\right]  
    + \varepsilon H(\bm{\Pi})   \label{wasserstein1_a}\vspace{1ex}\\
    \text { s.t. }  &\sum_{j=1}^N w_{i j}=1,\   \sum_{i=1}^M p_iw_{i j} =\sum_{i=1}^M \Pi_{i j} , \label{wasserstein1_b}\vspace{1ex}\\
    & \sum_{j=1}^N \Pi_{i j}= p_i,\quad \sum_{i=1}^M \sum_{j=1}^N \Pi_{i j} c_{i j} \leq P \label{wasserstein1_c}\vspace{1ex}\\
    &\sum_{i=1}^M \sum_{j=1}^N w_{i j} p_i d_{i j} \leq D,\label{wasserstein1_d}\vspace{1ex}
\end{align}
\end{subequations}
by introducing the multipliers $\bm{\alpha},\bm{\hat{\alpha}},\bm{\beta},\lambda\geq 0,\gamma\geq 0$, the Lagrangian function can be written as:
\begin{multline}\label{Lagrangian_w}
    \mathcal{L} \left(\bm{w}, \bm{\Pi}; \bm{\alpha},\bm{\hat{\alpha}},\bm{\beta},\lambda,\gamma \right) =
    \sum_{i=1}^M \sum_{j=1}^N w_{i j} p_i \log \frac{w_{i j}}{r_j}
    + \varepsilon \sum_{i=1}^M \sum_{j=1}^N \Pi_{i j} \log(\Pi_{i j}) \\
    + \sum_{i=1}^M  \alpha_i \left(\sum_{j=1}^N w_{i j} - 1\right)
    + \sum_{i=1}^M \hat{\alpha}_i \left(\sum_{j=1}^N \Pi_{i j} - p_i\right) + \lambda \left(\sum_{i=1}^M \sum_{j=1}^N c_{i j} \Pi_{i j} - P\right) \\
    + \gamma \left(\sum_{i=1}^M \sum_{j=1}^N p_i w_{i j} d_{i j} - D\right)
    + \sum_{j=1}^N \beta_j \left(\sum_{i=1}^M p_i w_{i j} - \sum_{i=1}^M \Pi_{i j}\right)
\end{multline}
Taking the derivative of $\mathcal{L} \left(\bm{w}, \bm{\Pi}; \bm{\alpha},\bm{\hat{\alpha}},\bm{\beta},\lambda,\gamma \right)$ with respect to the primal variable $\bdw$, one obtains
\begin{equation*} 
\begin{aligned}
    \frac{\partial \mathcal{L}}{\partial w_{i j}}=p_i(1+\log w_{i j}-\log r_j)+\alpha_{i}+\beta_{j}p_i
    +\gamma p_i d_{ij},
\end{aligned}
\end{equation*}
and it further yields the representation of $\bdw$ by dual variables
\begin{align}
        w_{ij} = r_je^{\beta_j-\frac{{\alpha}_i}{p_i}-\gamma d_{ij}-1} \label{w_w}
    \end{align}

Taking the derivative of $\mathcal{L} \left(\bm{w}, \bm{\Pi}; \bm{\alpha},\bm{\hat{\alpha}},\bm{\beta},\lambda,\gamma \right)$ with respect to the primal variable $\bm{\Pi}$, one obtains
\begin{equation*} 
\begin{aligned}
    \frac{\partial \mathcal{L}}{\partial \Pi_{i j}}=\varepsilon(1+\log \Pi_{i j})+\hat{\alpha}_{i}-\beta_{j}
    +\lambda c_{ij},
\end{aligned}
\end{equation*}
and it further yields the representation of $\bdw$ by dual variables
\begin{align}
        \Pi_{ij} = e^{-\frac{\beta_j}{\epsilon}-\frac{\hat{\alpha}_i}{\varepsilon }-\frac{\lambda c_{ij}}{\epsilon} -1} \label{Pi_w}
    \end{align}

By substituting \eqref{w_w} into the constraint $\sum_j w_{ij}=1$, one obtains the following way to update $\bm{\alpha}$
$$\alpha_i = p_i\log\left(\sum_{j=1}^Nr_je^{\beta_j-\gamma d_{ij}-1}\right)$$

By substituting \eqref{Pi_w} into the constraint $\sum_j \Pi_{ij}=p_i$, one obtains the following way to update $\bm{\hat{\alpha}}$
$$\hat{\alpha}_i = \varepsilon \log\left(\sum_{j=1}^N e^{-\frac{\beta_j}{\epsilon}-\frac{\lambda c_{ij}}{\epsilon} -1}\right)$$

By substituting \eqref{w_w} and \eqref{Pi_w} into the constraint $\sum_i p_i w_{ij}=\sum_i \Pi_{ij}$, $\bm{\beta}$ can be updated by the following way
$$\beta_j = \frac{\varepsilon}{\varepsilon+1}\log\left(\frac{\sum_{i=1}^M e^{-\frac{\lambda c_{ij}}{\epsilon}-\frac{\hat{\alpha}_i}{\varepsilon } -1}}{\sum_{i=1}^M p_ir_je^{-\gamma d_{ij}-\frac{{\alpha_i}}{p_i}-1}}\right)$$

By substituting \eqref{w_w} into the inequality constraint \eqref{wasserstein1_d}, one obtains the following way to update ${\gamma}$:\\
 find the root of the monotonic function $G(\gamma)$ using Newton's method:
    $$G(\gamma) = \sum_{i=1}^M{\sum_{j=1}^N} p_i r_{j}e^{\beta_j-\frac{{\alpha}_i}{p_i}-\gamma d_{ij}-1} d_{i j}-D$$
the monotonic property is due to 
$$G^{\prime}(\gamma) = -\sum_{i=1}^M{\sum_{j=1}^N} p_i r_{j}e^{\beta_j-\frac{{\alpha}_i}{p_i}-\gamma d_{ij}-1} d_{i j}^2\leq 0$$
    
By substituting \eqref{Pi_w} into the inequality constraint \eqref{wasserstein1_c}, one obtains the following way to update ${\lambda}$:\\
find the root of the monotonic function $F(\lambda)$ using Newton's method:
    $$F(\lambda) = \sum_{i=1}^M{\sum_{j=1}^N}  e^{-\frac{\beta_j}{\epsilon}-\frac{\hat{\alpha}_i}{\varepsilon }-\frac{\lambda c_{ij}}{\epsilon} -1} c_{i j}-P$$
the monotonic property is due to 
$$F^{\prime}(\lambda) = -\sum_{i=1}^M{\sum_{j=1}^N}  e^{-\frac{\beta_j}{\epsilon}-\frac{\hat{\alpha}_i}{\varepsilon }-\frac{\lambda c_{ij}}{\epsilon} -1} c_{i j}^2/\varepsilon\leq 0$$
The \alg algorithm for Wasserstein case can be summarized as follow:
\begin{algorithm}[ht]
	\caption{\alg algorithm for Wasserstein case}
	\label{alg:OT_ibp}
	\begin{algorithmic}[1]
		\REQUIRE Distribution $p_{i}$, distortion measure $d_{ij}$, cost $c_{ij}$, the entropy regularization constant $\varepsilon$,
        maximum iteration number $max\_iter$.
		\ENSURE Minimal value $\sum_{i=1}^{M} \sum_{j=1}^{N} w_{i j}p_i(\log w_{i j}-\log r_j)$.
        \STATE \textbf{Initialization:} $\bm{r}=\frac{1}{N}\mathbf{1}_{N},\bm{\phi}=\mathbf{1}_{M},\bm{\hat{\phi}}=\mathbf{1}_{M},\bm{\beta}=\mathbf{0}_{N},\gamma=0,\lambda=0$
		\FOR{$\ell = 1 : max\_iter$}
		\STATE $r_{j} \gets \sum_i p_i w_{ij}$
		\WHILE{the inner iteration is sufficient}
            \STATE $\phi_i\gets 1/\sum_j r_j e^{\beta_j-\gamma d_{ij}}$
		\STATE $\hat{\phi}_i\gets 1/\sum_j e^{-\beta_j/\varepsilon-\lambda c_{ij}/\varepsilon}$
		\STATE $\beta_j\gets \frac{\varepsilon}{\varepsilon+1}\log\left(\frac{\sum_{i=1}^M \exp({-\frac{\lambda c_{ij}}{\epsilon}})\hat{\phi}_i}{\sum_{i=1}^M p_ir_je^{-\gamma d_{ij}}\phi_i}\right)$
            \STATE Find the root of $G(\gamma)$ by Newton's method in $\mathbb{R}^{+}$
            \STATE Find the root of $F(\lambda)$ by Newton's method in $\mathbb{R}^{+}$
            \ENDWHILE
            \STATE $w_{ij} \gets r_je^{\beta_j-\gamma d_{ij}}\phi_i$
		\ENDFOR
		\STATE \textbf{end}
		\RETURN $\sum_{i=1}^{M} \sum_{j=1}^{N} w_{i j}p_i(\log w_{i j}-\log r_j)$
	\end{algorithmic}
\end{algorithm}\\
Here, we substitute $e^{-\alpha_i/p_i-1}$ by $\phi_i$ and substitute $e^{-\hat{\alpha}_i/\varepsilon-1}$ by $\hat{\phi}_i$ for convenience.

\section{Convergence Analysis} \label{sec:convergence}

In the following subsections, we rigorously give the convergence analysis of the \alg algorithm under both KL divergence and Wasserstein metric cases.

\subsection{KL divergence case}
Since the algorithm proceeds with both outer and inner iterations, we establish the convergence of each loop separately. In order to prove the convergence of the inner iterations, we begin by establishing the following lemma.

\begin{lemma}\label{lem-0}
\begin{multline*}
    F_1(\bm{a},\bm{b},\lambda,\gamma) =
    \sum_{i=1}^M \sum_{j=1}^N p_i r_j e^{a_j - b_i/p_i - 1 - \lambda d_{ij}}
    - \gamma \sum_{j=1}^N p_j \log a_j \\
    + \sum_{i=1}^M b_i
    + \gamma \log \gamma
    + (P - 1)\gamma
    + D\lambda
\end{multline*}
is convex over $\bm{a},\bm{b},\lambda$ and $\gamma$.
\end{lemma}
\begin{proof}
    We first show $e^{a_j-b_i/p_i-1-\lambda d_{ij}}$ is convex with respect to $a_j,b_i,\lambda$. By taking derivatives trice, we can obtain the Hessian matrix as follow:
    $$e^{a_j-b_i/p_i-1-\lambda d_{ij}}\left(
    \begin{array}{ccc}
       1  & -1/p_i &-d_{ij} \\
        -1/p_i & 1/p_i^2 & d_{ij}/p_i \\
         -d_{ij} &d_{ij}/p_i& d_{ij}^2
    \end{array}\right)$$ 
    Since
    \begin{equation*}
    \begin{aligned}
    \left(
    \begin{array}{ccc}
       1  & -1/p_i &-d_{ij} \\
        -1/p_i & 1/p_i^2 & d_{ij}/p_i \\
         -d_{ij} &d_{ij}/p_i& d_{ij}^2
    \end{array}\right) =(1,-1/p_i,-d_{ij})^T (1,-1/p_i,-d_{ij}),
    \end{aligned}
    \end{equation*}
    we have the Hessian matrix is positive semidefinite. Then $e^{a_j-b_i/p_i-1-\lambda d_{ij}}$ is convex with respect to $a_j,b_i,\lambda$, and is convex with respect to $\bm{a},\bm{b},\lambda,\gamma$ further.

    Since the sum of convex functions is convex, we have $\sum_{i=1}^M\sum_{j=1}^N p_ir_je^{a_j-b_i/p_i-1-\lambda d_{ij}}$ is convex with respect to $\bm{a},\bm{b},\lambda,\gamma$.

    Then, we only need to prove the convexity of $-\gamma\sum_{j=1}^N p_j\log{a_j}+\gamma\log\gamma$. Note that the function $\gamma\log\gamma-\gamma \log a_j$ is convex with respect to $\gamma, a_j$, since the Hessian matrix $$\left(
    \begin{array}{cc}
       1/\gamma  &  -1/a_j\\
        -1/a_j & \gamma /a_j^2 
    \end{array}\right)$$ is positive semi-definite.\\
    \begin{equation*}
            -\gamma\sum_{j=1}^N p_j\log{a_j}+\gamma\log\gamma =\sum_{j=1}^N p_j(\gamma\log\gamma-\gamma \log a_j)
    \end{equation*}
    is the sum of convex function, so it is convex.
    Finally, we proved the convexity of $F_1(\bm{a},\bm{b},\lambda,\gamma)$
    
\end{proof}

The convergence of the inner iterations is shown in the Theorem \ref{thm-1}.
\begin{theorem} \label{thm-1}
The inner iteration is convergent.
\end{theorem}
\begin{proof}
Since $F_1(\bm{a},\bm{b},\lambda,\gamma)$ is convex in Lemma \ref{lem-0}, we can obtain the following alternating minimization
for $F(\bm{a},\bm{b},\lambda,\gamma)$ is convergent.
\begin{align*}
    \bm{a}^{n+1},\bm{b}^{n+1} &= \arg \min_{\bm{a},\bm{b}} F(\bm{a},\bm{b},\lambda^n,\gamma^n),\\
    \lambda^{n+1},\gamma^{n+1} &= \arg \min_{\lambda,\gamma} F(\bm{a}^{n+1},\bm{b}^{n+1},\lambda,\gamma),
\end{align*}
and in the first block we can have an inner iteration:
\begin{align*}
    \bm{b}^{n+1,m+1} &= \arg \min_{\bm{b}} F(\bm{a}^{n+1,m},\bm{b},\lambda^n,\gamma^n),\\
    \bm{a}^{n+1,m+1} &= \arg \min_{\bm{a}} F(\bm{a},\bm{b}^{n+1,m+1},\lambda^n,\gamma^n), 
\end{align*}
due to the convexity the first minimization can be solved by the first-order condition: 
\begin{align*}
    \sum_{j=1}^N p_ir_je^{a_j^{n+1,m}+b_i-\lambda d_{ij}} - p_i = 0,
\end{align*}
thus $b_i^{n+1,m+1} = \log(-\sum_{j=1}^N r_je^{a_j^{n+1,m}-\lambda d_{ij}})$.

Similarly $a_j^{n+1,m+1}$ can be solved by solving 
\begin{equation*}
    \sum_{i=1}^M p_ir_je^{a^{n+1,m+1}_j+b^{n+1,m+1}_i-\lambda d_{ij}}-\gamma \frac{p_j}{a^{n+1,m+1}_j} = 0,
\end{equation*}
and can be converted to find the root of  
\begin{align*}
    f_j(a_j) = \gamma p_j\frac{e^{-a_j}}{a_{j}} - {\sum_{i=1}^M{p_ir_je^{b^{n+1}_i-\lambda d_{ij}}}}
\end{align*}
by Newton's method due to its monotony on $\mathbb{R}^+$.

For the second block, we can directly find the minimal solutions separately since $F(\bm{a},\bm{b},\lambda,\gamma)= F_1(\bm{a},\bm{b},\lambda)+F_2(\bm{a},\bm{b},\gamma)$. For $\gamma^{n+1}$ we have:
\begin{equation*}
\gamma^{n+1} = e^{\sum_{i=1}^M p_i \log a^{n+1}_{i}-P}.
\end{equation*}
As for $\lambda^{n+1}$, we can use Newton's method to find the root of
\begin{equation*}
g(\lambda) = \sum_{i}^M {p_i}{\sum_{j=1}^N r_jd_{ij}e^{a^{n+1}_j+b^{n+1}_i-\lambda d_{ij}}}-D.
\end{equation*}

By setting $c_j = e^{-b_j}$, we can draw a conclusion the above iteration method is convergent.
\end{proof}

Next, to establish the convergence of the outer iterations, we first prove the following lemmas.
\begin{lemma} \label{lem-1}
For any $\boldsymbol{w}$ satisfying $\sum_{i=1}^M \sum_{j=1}^N w_{i j} p_i d_{i j} = D,\sum_{j=1}^N p_j\log(\sum_{i=1}^M p_i w_{ij})=T$ and setting $f(\boldsymbol{w}, \boldsymbol{r}) = \sum_{i=1}^M \sum_{j=1}^N\left(w_{i j} p_i\right)\left[\log w_{i j}-\log r_j\right]$, we have
$$
f(\boldsymbol{w}, \boldsymbol{r})-f(\tilde{\boldsymbol{w}}(\boldsymbol{r}), \boldsymbol{r}) \geq D\left(\sum_{i=1}^M p_i \bm{w}_i \| \sum_{i=1}^M p_i \tilde{\bm{w}}_i\right),
$$
where $\tilde{\boldsymbol{w}}(\boldsymbol{r})$ is the optimal solution to problem \eqref{eq3}.
\end{lemma}
\begin{proof}
Assume $a_j,b_i,\lambda,\gamma$ is the optimal multiplier of problem \eqref{eq3}, then $\tilde{w}_{ij}=r_je^{a_j-\lambda d_{ij}-b_i/p_i-1}$.
\begin{align*}
& f(\boldsymbol{w}, \boldsymbol{r}) - f(\tilde{\boldsymbol{w}}(\boldsymbol{r}), \boldsymbol{r})  \\
&= f(\boldsymbol{w}, \boldsymbol{r}) + \lambda D - \gamma T - f(\tilde{\boldsymbol{w}}, \boldsymbol{r}) - \lambda D + \gamma T \\
&= \left( f(\boldsymbol{w}, \boldsymbol{r}) + \lambda \sum_{i,j} p_i w_{ij} d_{ij}
   - \gamma \sum_{j} p_j \log\left(\sum_i p_i w_{ij}\right) \right) \\
&\quad - \left( f(\tilde{\boldsymbol{w}}, \boldsymbol{r}) + \lambda \sum_{i,j} p_i \tilde{w}_{ij} d_{ij}
   - \gamma \sum_{j} p_j \log\left(\sum_i p_i \tilde{w}_{ij}\right) \right) \\
&= \left( \sum_{i,j} p_i w_{ij} \log \frac{w_{ij}}{r_j e^{a_j - \lambda d_{ij}}}
   + \sum_{i,j} a_j p_i w_{ij}
   - \gamma \sum_{j} p_j \log\left(\sum_i p_i w_{ij}\right) \right) \\
&\quad - \left( \sum_{i,j} p_i \tilde{w}_{ij} \log \left( e^{a_j} e^{-b_i/p_i - 1} \right)
   - \gamma \sum_{j} p_j \log\left(\sum_i p_i \tilde{w}_{ij}\right) \right) \\
&= \sum_{i,j} p_i w_{ij} \log \frac{w_{ij}}{r_j e^{a_j - \lambda d_{ij}}}
   - \sum_{i} \left( \sum_j p_i \tilde{w}_{ij} \right) \log e^{-b_i/p_i - 1} \\
&\quad + \sum_{i,j} a_j (p_i w_{ij} - p_i \tilde{w}_{ij})
   - \gamma \sum_j p_j \log\left( \frac{\sum_i p_i w_{ij}}{\sum_i p_i \tilde{w}_{ij}} \right) \\
&= \sum_{i,j} p_i w_{ij} \log \frac{w_{ij}}{r_j e^{a_j - \lambda d_{ij}}}
   - \sum_{i} \left( \sum_j p_i w_{ij} \right) \log e^{-b_i/p_i - 1} \\
&\quad + \sum_{i,j} \frac{\gamma p_j}{\sum_i p_i \tilde{w}_{ij}} (p_i w_{ij} - p_i \tilde{w}_{ij})
   - \gamma \sum_j p_j \log\left( \frac{\sum_i p_i w_{ij}}{\sum_i p_i \tilde{w}_{ij}} \right) \\
&= \gamma \left( \sum_j p_j \frac{\sum_i p_i w_{ij}}{\sum_i p_i \tilde{w}_{ij}}
   - \sum_j p_j \log\left( \frac{\sum_i p_i w_{ij}}{\sum_i p_i \tilde{w}_{ij}} \right) - 1 \right) \\
&\quad + \sum_{i,j} p_i w_{ij} \log \frac{w_{ij}}{r_j e^{a_j - \lambda d_{ij}} e^{-b_i/p_i - 1}} \\
&\geq \sum_{i=1}^M p_i D\left(\boldsymbol{w}_i \,\|\, \tilde{\boldsymbol{w}}_i\right)
   \geq D\left(\sum_{i=1}^M p_i \boldsymbol{w}_i \,\|\, \sum_{i=1}^M p_i \tilde{\boldsymbol{w}}_i\right).
\end{align*}
here the first inequality we use the fact that $h(x) = x-1-\log(x)\geq 0$, thus

\begin{align*}
\sum_j p_j \frac{\sum_i p_i w_{ij}}{\sum_i p_i \hat{w}_{ij}}
- \sum_j p_j \log\!\left( \frac{\sum_i p_i w_{ij}}{\sum_i p_i \hat{w}_{ij}} \right) - 1 
= \sum_j p_j \, h\!\left( \frac{\sum_i p_i w_{ij}}{\sum_i p_i \hat{w}_{ij}} \right) 
\geq 0.
\end{align*}
The second inequality follows from the joint convexity of the KL divergence.

\end{proof}

\begin{lemma} \label{lem-2}
$$
f(\boldsymbol{w}, \boldsymbol{r})-f(\boldsymbol{w}, \tilde{\boldsymbol{r}}(\boldsymbol{w})) = D\left(\tilde{\boldsymbol{r}}(\boldsymbol{w}) \| \boldsymbol{r}\right),
$$ 
where $\tilde{\boldsymbol{r}}(\boldsymbol{w}) = \sum_{i=1}^N  p_iw_{i j}$, which is the minimizer.
\end{lemma}

\begin{proof}
\begin{multline*}
f(\boldsymbol{w}, \boldsymbol{r}) - f(\boldsymbol{w}, \tilde{\boldsymbol{r}}(\boldsymbol{w})) \\
= \sum_{i, j} w_{ij} p_i \left[ \log w_{ij} - \log r_j \right]
 - \sum_{i, j} w_{ij} p_i \left[ \log w_{ij} - \log \tilde{r}_j \right] \\
= \sum_{i, j} w_{ij} p_i \left( \log \tilde{r}_j - \log r_j \right)
= \sum_j \tilde{r}_j \log \frac{\tilde{r}_j}{r_j} 
= D\big(\tilde{\boldsymbol{r}}(\boldsymbol{w}) \,\|\, \boldsymbol{r}\big).
\end{multline*}
\end{proof}

\begin{lemma} \label{lem-3}
Following the iteration rule: $\bm{r}^n=\tilde{\boldsymbol{r}}(\boldsymbol{w}^{n-1}),$ $\bm{w}^n=\tilde{\boldsymbol{w}}(\boldsymbol{r}^{n}),$ we have
$$
f(\boldsymbol{w}^n, \boldsymbol{r}^n) \leq f(\boldsymbol{w}^{n-1}, \boldsymbol{r}^{n-1}).
$$ 
\end{lemma}

\begin{proof}
\begin{multline*}
f(\boldsymbol{w}^n, \boldsymbol{r}^n) - f(\boldsymbol{w}^{n-1}, \boldsymbol{r}^{n-1}) \\
= \left( f(\boldsymbol{w}^n, \boldsymbol{r}^n) - f(\boldsymbol{w}^{n-1}, \boldsymbol{r}^{n}) \right)
   + \left( f(\boldsymbol{w}^{n-1}, \boldsymbol{r}^{n}) - f(\boldsymbol{w}^{n-1}, \boldsymbol{r}^{n-1}) \right) \\
\leq -D\!\left(\sum_{i=1}^M p_i \boldsymbol{w}^n_i \,\big\|\, \sum_{i=1}^M p_i \boldsymbol{w}^{n+1}_i\right)
     -D\!\left(\boldsymbol{r}^{n} \,\|\, \boldsymbol{r}^{n-1}\right) \\
\leq 0.
\end{multline*}
Here, the inequality follows from Lemma \ref{lem-1} and \ref{lem-2}.
\end{proof}

\begin{theorem}
    Suppose $(\bm{w}^*,\bm{r}^*)$ is the global optimizer of \eqref{eq2}, and we have
    $$f\left(\boldsymbol{w}^{n}, \boldsymbol{r}^{n}\right)-f\left(\boldsymbol{w}^*, \boldsymbol{r}^*\right) \simeq O(1 / n)$$
\end{theorem}

\begin{proof}
Noticing $\bm{r}^* = \tilde{\boldsymbol{r}}(\boldsymbol{w}^*)$ and $\bm{r}^{n+1} = \tilde{\boldsymbol{r}}(\boldsymbol{w}^n)$, also under Lemma \ref{lem-1} and \ref{lem-2}, we have
\begin{multline*}
f\!\left(\boldsymbol{w}^{n}, \boldsymbol{r}^{n}\right)
- f\!\left(\boldsymbol{w}^*, \boldsymbol{r}^*\right) \\
= \left( f\!\left(\boldsymbol{w}^{n}, \boldsymbol{r}^{n}\right)
      - f\!\left(\boldsymbol{w}^*, \boldsymbol{r}^{n}\right) \right)
  + \left( f\!\left(\boldsymbol{w}^*, \boldsymbol{r}^{n}\right)
      - f\!\left(\boldsymbol{w}^*, \boldsymbol{r}^*\right) \right) \\
\leq -D\!\left(\sum_i p_i \boldsymbol{w}_i^*
        \,\big\|\, \sum_i p_i \boldsymbol{w}_i^{n}\right)
     + D\!\left(\boldsymbol{r}^* \,\|\, \boldsymbol{r}^{n}\right) \\
= -D\!\left(\boldsymbol{r}^* \,\|\, \boldsymbol{r}^{n+1}\right)
   + D\!\left(\boldsymbol{r}^* \,\|\, \boldsymbol{r}^{n}\right)
\end{multline*}
for $n=1,2, \cdots$. Then, by summing up these inequalities, we have
\begin{multline*}
\sum_{k=1}^n\left(f\left(\boldsymbol{w}^{k}, \boldsymbol{r}^{k}\right)-f\left(\boldsymbol{w}^*, \boldsymbol{r}^*\right)\right) 
\\ \leq D\left(\boldsymbol{r}^* \| \boldsymbol{r}^{1}\right)-D\left(\boldsymbol{r}^* \| \boldsymbol{r}^{n+1}\right) \leq D\left(\boldsymbol{r}^* \| \boldsymbol{r}^{1}\right) \triangleq C,
\end{multline*}
where $C$ is a constant. 

According to Lemma \ref{lem-3}, we obtain
\begin{equation*}
\begin{aligned}
0 \leq n\left(f\left(\boldsymbol{w}^{n}, \boldsymbol{r}^{n}\right)-f\left(\boldsymbol{w}^*, \boldsymbol{r}^*\right)\right) \leq \sum_{k=1}^n\left(f\left(\boldsymbol{w}^{k}, \boldsymbol{r}^{k}\right)-f\left(\boldsymbol{w}^*, \boldsymbol{r}^*\right)\right) \leq C,
\end{aligned}
\end{equation*}
i.e., $f\left(\boldsymbol{w}^{n}, \boldsymbol{r}^{n}\right)-f\left(\boldsymbol{w}^*, \boldsymbol{r}^*\right) \simeq O(1 / n)$.
\end{proof}

\subsection{Wasserstein metric case}

Similar to the KL case, we first establish the convergence of the inner iterations.
\begin{lemma}\label{lem-5}
\begin{multline*}
F_2(\bm{\alpha},\bm{\hat{\alpha}},\bm{\beta},\lambda,\gamma) =
\sum_{i=1}^M \sum_{j=1}^N p_i r_{j}
    e^{\beta_j - \frac{\alpha_i}{p_i} - \gamma d_{ij} - 1}
+ \varepsilon \sum_{i=1}^M \sum_{j=1}^N
    e^{-\frac{\beta_j}{\varepsilon} - \frac{\hat{\alpha}_i}{\varepsilon}
      - \frac{\lambda c_{ij}}{\varepsilon} - 1} \\
+ P\lambda + D\gamma
+ \sum_{i=1}^M \alpha_i
+ \sum_{i=1}^M \hat{\alpha}_i
\end{multline*}
is convex over $\bm{\alpha},\bm{\hat{\alpha}},\bm{\beta}, \lambda$ and $\gamma$.
\end{lemma}
\begin{proof}
    In Lemma \ref{lem-0}, we proved the convexity of $e^{\beta_j-\frac{{\alpha}_i}{p_i}-\gamma d_{ij}-1}$ and $e^{-\frac{\beta_j}{\epsilon}-\frac{\hat{\alpha}_i}{\varepsilon }-\frac{\lambda c_{ij}}{\epsilon} -1}$. Since the sum of convex function is convex, we have $F_2(\bm{\alpha},\bm{\hat{\alpha}},\bm{\beta},\lambda,\gamma)$ is convex with respect to $\bm{\alpha},\bm{\hat{\alpha}},\bm{\beta}, \lambda$ and $\gamma$.
\end{proof}
\begin{theorem}
The inner iteration method is convergent.
\end{theorem}

\begin{proof}
    Same to the proof of Theorem \ref{thm-1}, the iteration method coincides with alternatively minimizing the convex function $F_2(\bm{\alpha},\bm{\hat{\alpha}},\bm{\beta},\lambda,\gamma)$, so it is convergent.
\end{proof}

Next, we will prove the convergence of the outer iteration. 
\begin{lemma} \label{lem-6}
For any $\boldsymbol{w}$ satisfying $\sum_{i=1}^M \sum_{j=1}^N w_{i j} p_i d_{i j} = D$ and any $\boldsymbol{\Pi}$ satisfying $\sum_{i=1}^M \sum_{j=1}^N \Pi_{i j}  c_{i j} = P$ and $\sum_{i=1}^M p_i w_{ij}=\sum_{i=1}^M\Pi_{ij}$ and setting $
g(\boldsymbol{w}, \boldsymbol{\Pi}, \boldsymbol{r}) = \sum_{i=1}^M \sum_{j=1}^N\left(w_{i j} p_i\right)\left[\log w_{i j}-\log r_j\right] + \varepsilon H(\bm{\Pi}),$
we have
\begin{multline*}
g(\boldsymbol{w}, \boldsymbol{\Pi}, \boldsymbol{r})
 - g(\tilde{\boldsymbol{w}}(\boldsymbol{r}),
      \tilde{\boldsymbol{\Pi}}(\boldsymbol{r}),
      \boldsymbol{r}) \\
\geq D\!\left(\sum_{i=1}^M p_i \boldsymbol{w}_i
        \,\big\|\, \sum_{i=1}^M p_i \tilde{\boldsymbol{w}}_i\right)
   + \varepsilon\, D\!\left(\sum_{i=1}^M \boldsymbol{\Pi}_i
        \,\big\|\, \sum_{i=1}^M \tilde{\boldsymbol{\Pi}}_i\right)
\end{multline*}

where the expression of $\tilde{\boldsymbol{w}}(\boldsymbol{r})$ is \eqref{w_w} and the expression of $\tilde{\boldsymbol{\Pi}}(\boldsymbol{r})$ is \eqref{Pi_w}, and $\bm{\alpha},\bm{\hat{\alpha}},\bm{\beta},\lambda,\gamma$ are the optimal solution of the inner iteration and are related to $\bm{r}$.
\end{lemma}

\begin{proof}
    Assume $\bm{\alpha},\bm{\hat{\alpha}},\bm{\beta},\lambda,\gamma$ are the optimal solution of the inner iteration, then $\tilde{w}_{ij}= r_je^{\beta_j-\frac{{\alpha}_i}{p_i}-\gamma d_{ij}-1}$ and $\tilde{\Pi_{ij}} = e^{-\frac{\beta_j}{\epsilon}-\frac{\hat{\alpha}_i}{\varepsilon }-\frac{\lambda c_{ij}}{\epsilon} -1}$, and we have
\begin{equation*}
\renewcommand{\arraystretch}{1.4}
\begin{aligned}
& g(\boldsymbol{w}, \boldsymbol{\Pi}, \boldsymbol{r})
 - g(\tilde{\boldsymbol{w}}(\boldsymbol{r}), \tilde{\boldsymbol{\Pi}}(\boldsymbol{r}), \boldsymbol{r}) \\
&= g(\boldsymbol{w}, \boldsymbol{\Pi}, \boldsymbol{r}) + \gamma D + \lambda P
 - g(\tilde{\boldsymbol{w}}(\boldsymbol{r}), \tilde{\boldsymbol{\Pi}}(\boldsymbol{r}), \boldsymbol{r}) - \gamma D - \lambda P \\
&= g(\boldsymbol{w}, \boldsymbol{\Pi}, \boldsymbol{r})
   + \gamma \sum_{i,j} p_i w_{ij} d_{ij}
   + \lambda \sum_{i,j} c_{ij} \Pi_{ij} \\
&\quad - g(\tilde{\boldsymbol{w}}(\boldsymbol{r}), \tilde{\boldsymbol{\Pi}}(\boldsymbol{r}), \boldsymbol{r})
   - \gamma \sum_{i,j} p_i \tilde{w}_{ij} d_{ij}
   - \lambda \sum_{i,j} c_{ij} \tilde{\Pi}_{ij} \\
&= \sum_{i,j} p_i w_{ij} \log \frac{w_{ij}}{r_j e^{-\gamma d_{ij}}}
   + \varepsilon \sum_{i,j} \Pi_{ij} \log \frac{\Pi_{ij}}{e^{-\lambda c_{ij}/\varepsilon}} \\
&\quad - \sum_{i,j} p_i \tilde{w}_{ij} \log\!\left(e^{\beta_j} e^{-\alpha_i/p_i - 1}\right)
   - \varepsilon \sum_{i,j} \tilde{\Pi}_{ij} \log\!\left(e^{-\beta_j/\varepsilon - \hat{\alpha}/\varepsilon - 1}\right) \\
&= \sum_{i,j} p_i w_{ij} \log \frac{w_{ij}}{r_j e^{-\gamma d_{ij}}}
   + \varepsilon \sum_{i,j} \Pi_{ij} \log \frac{\Pi_{ij}}{e^{-\lambda c_{ij}/\varepsilon}} \\
&\quad - \sum_i p_i \left(\sum_j \tilde{w}_{ij}\right) \log\!\left(e^{-\alpha_i/p_i - 1}\right)
   - \varepsilon \sum_i \left(\sum_j \tilde{\Pi}_{ij}\right) \log\!\left(e^{-\hat{\alpha}_i/\varepsilon - 1}\right) \\
&\quad \quad - \sum_j \beta_j \left(\sum_i p_i \tilde{w}_{ij} - \sum_i \tilde{\Pi}_{ij}\right) \\
&= \sum_{i,j} p_i w_{ij} \log \frac{w_{ij}}{r_j e^{-\gamma d_{ij}}}
   + \varepsilon \sum_{i,j} \Pi_{ij} \log \frac{\Pi_{ij}}{e^{-\lambda c_{ij}/\varepsilon}} \\
&\quad - \sum_i p_i \left(\sum_j w_{ij}\right) \log\!\left(e^{-\alpha_i/p_i - 1}\right)
   - \varepsilon \sum_i \left(\sum_j \Pi_{ij}\right) \log\!\left(e^{-\hat{\alpha}_i/\varepsilon - 1}\right) \\
&= \sum_{i,j} p_i w_{ij} \log \frac{w_{ij}}{r_j e^{-\gamma d_{ij} - \alpha_i/p_i - 1}}
   + \varepsilon \sum_{i,j} \Pi_{ij} \log \frac{\Pi_{ij}}{e^{-\lambda c_{ij}/\varepsilon - \hat{\alpha}_i/\varepsilon - 1}} \\
&= \sum_{i,j} p_i w_{ij} \log \frac{w_{ij}}{r_j e^{-\gamma d_{ij} - \alpha_i/p_i - 1 + \beta_j}}
   + \varepsilon \sum_{i,j} \Pi_{ij} \log \frac{\Pi_{ij}}{e^{-\lambda c_{ij}/\varepsilon - (\hat{\alpha}_i + \beta_j)/\varepsilon - 1}} \\
&= \sum_{i,j} p_i w_{ij} \log \frac{w_{ij}}{\tilde{w}_{ij}}
   + \varepsilon \sum_{i,j} \Pi_{ij} \log \frac{\Pi_{ij}}{\tilde{\Pi}_{ij}} \\
&= \sum_{i=1}^M p_i D\!\left(\boldsymbol{w}_i \,\|\, \tilde{\boldsymbol{w}}_i\right)
   + \varepsilon D\!\left(\boldsymbol{\Pi} \,\|\, \tilde{\boldsymbol{\Pi}}\right) \\
&\geq D\!\left(\sum_{i=1}^M p_i \boldsymbol{w}_i \,\big\|\, \sum_{i=1}^M p_i \tilde{\boldsymbol{w}}_i\right)
   + \varepsilon D\!\left(\sum_{i=1}^M \boldsymbol{\Pi}_i \,\big\|\, \sum_{i=1}^M \tilde{\boldsymbol{\Pi}}_i\right).
\end{aligned}
\end{equation*}
Here the inequality follows from the joint convexity and log-sum inequality of the KL divergence.

\end{proof}

\begin{lemma}  \label{lem-7}
$$
g(\boldsymbol{w}, \boldsymbol{\Pi}, \boldsymbol{r})-g(\boldsymbol{w}, \boldsymbol{\Pi}, \tilde{\boldsymbol{r}}(\boldsymbol{w})) = D\left(\tilde{\boldsymbol{r}}(\boldsymbol{w}) \| \boldsymbol{r}\right),
$$ 
where $\tilde{\boldsymbol{r}}(\boldsymbol{w}) = \bm{r} = \sum_{i=1}^N  p_iw_{i j}.$
\end{lemma}

\begin{proof}
\begin{multline*}
g(\boldsymbol{w}, \boldsymbol{\Pi}, \boldsymbol{r})
 - g(\boldsymbol{w}, \boldsymbol{\Pi}, \tilde{\boldsymbol{r}}(\boldsymbol{w})) \\
= \sum_{i,j} w_{ij} p_i \left[\log w_{ij} - \log r_j \right]
 - \sum_{i,j} w_{ij} p_i \left[\log w_{ij} - \log \tilde{r}_j \right] \\
= \sum_{i,j} w_{ij} p_i \left( \log \tilde{r}_j - \log r_j \right) \\
= \sum_j \tilde{r}_j \log \frac{\tilde{r}_j}{r_j}
= D\!\left(\tilde{\boldsymbol{r}}(\boldsymbol{w}) \,\|\, \boldsymbol{r}\right)
\end{multline*}
\end{proof}

\begin{lemma}\label{lem-8}
Following the iteration rule: $\bm{r}^n=\tilde{\boldsymbol{r}}(\boldsymbol{w}^{n-1}),$ $\bm{w}^n=\tilde{\boldsymbol{w}}(\boldsymbol{r}^{n}),$ and $\bm{\Pi}^n=\tilde{\boldsymbol{\Pi}}(\boldsymbol{r}^{n})$, we have
$$
g(\boldsymbol{w}^n, \boldsymbol{\Pi}^n, \boldsymbol{r}^n) \leq g(\boldsymbol{w}^{n-1}, \boldsymbol{\Pi}^{n-1}, \boldsymbol{r}^{n-1}).
$$ 
\end{lemma}
\begin{proof}
Using the two Lemmas above, and noting that $\bm{r}^n=\tilde{\boldsymbol{r}}(\boldsymbol{w}^{n-1}),$ $\bm{w}^n=\tilde{\boldsymbol{w}}(\boldsymbol{r}^{n}),$ and $\bm{\Pi}^n=\tilde{\boldsymbol{\Pi}}(\boldsymbol{r}^{n})$, we have
\begin{multline*}
g(\boldsymbol{w}^n, \boldsymbol{\Pi}^n, \boldsymbol{r}^n)
 - g(\boldsymbol{w}^{n-1}, \boldsymbol{\Pi}^{\,n-1}, \boldsymbol{r}^{\,n-1})
\\= \big( g(\boldsymbol{w}^n, \boldsymbol{\Pi}^n, \boldsymbol{r}^n)
     - g(\boldsymbol{w}^{n-1}, \boldsymbol{\Pi}^{\,n-1}, \boldsymbol{r}^n) \big)  + \big( g(\boldsymbol{w}^{n-1}, \boldsymbol{\Pi}^{\,n-1}, \boldsymbol{r}^n)
     - g(\boldsymbol{w}^{n-1}, \boldsymbol{\Pi}^{\,n-1}, \boldsymbol{r}^{\,n-1}) \big) \\
\leq -D\!\left(\sum_{i=1}^M p_i \boldsymbol{w}_i^{\,n-1}
        \,\big\|\, \sum_{i=1}^M p_i \boldsymbol{w}_i^{\,n}\right)  - \varepsilon\, D\!\left(\sum_{i=1}^M \boldsymbol{\Pi}_i^{\,n-1}
        \,\big\|\, \sum_{i=1}^M \boldsymbol{\Pi}_i^{\,n}\right) - D(\boldsymbol{r}^{\,n} \,\|\, \boldsymbol{r}^{\,n-1}) \leq 0   
\end{multline*}
\end{proof}

Leveraging the above lemma, we arrive at the following proof of convergence for the outer iteration.

\begin{theorem}
    Suppose $(\bm{w}^*,\bm{\Pi}^*,\bm{r}^*)$ is the global optimizer of \eqref{wasserstein}, and we have
    $$g(\boldsymbol{w}^n, \boldsymbol{\Pi}^n, \boldsymbol{r}^n) -g(\boldsymbol{w}^*, \boldsymbol{\Pi}^*, \boldsymbol{r}^*)  \leq \frac{\varepsilon C}{(1+\varepsilon)^n-1}, $$
    which further shows that the convergence rate of the algorithm is $O(1/n)$.
\end{theorem}
\begin{proof}
Noticing $\bm{r}^* = \tilde{\boldsymbol{r}}(\boldsymbol{w}^*)$ and $\bm{r}^{n+1} = \tilde{\boldsymbol{r}}(\boldsymbol{w}^n)$, also under Lemma \ref{lem-6} and \ref{lem-7}, we have
\begin{multline*}
g\!\left(\boldsymbol{w}^{n}, \boldsymbol{\Pi}^{n}, \boldsymbol{r}^{n}\right)
 - g\!\left(\boldsymbol{w}^*, \boldsymbol{\Pi}^*, \boldsymbol{r}^*\right) \\
= \big( g\!\left(\boldsymbol{w}^{n}, \boldsymbol{\Pi}^{n}, \boldsymbol{r}^{n}\right)
     - g\!\left(\boldsymbol{w}^*, \boldsymbol{\Pi}^*, \boldsymbol{r}^{n}\right) \big) 
 + \big( g\!\left(\boldsymbol{w}^*, \boldsymbol{\Pi}^*, \boldsymbol{r}^{n}\right)
     - g\!\left(\boldsymbol{w}^*, \boldsymbol{\Pi}^*, \boldsymbol{r}^*\right) \big) \\
\leq - D\!\left(\sum_i p_i \boldsymbol{w}_i^*
        \,\big\|\, \sum_i p_i \boldsymbol{w}_i^{\,n}\right)
      - \varepsilon\, D\!\left(\sum_i \boldsymbol{\Pi}_i^*
        \,\big\|\, \sum_i \boldsymbol{\Pi}_i^{\,n}\right) 
 + D\!\left(\boldsymbol{r}^* \,\|\, \boldsymbol{r}^{\,n}\right) \\
= - (1+\varepsilon)\, D\!\left(\boldsymbol{r}^* \,\|\, \boldsymbol{r}^{\,n+1}\right)
  + D\!\left(\boldsymbol{r}^* \,\|\, \boldsymbol{r}^{\,n}\right)
\end{multline*}
for $n=1,2, \cdots$. Then, by summing up these inequalities, we have
\begin{multline*}
    \sum_{k=1}^n(1+\varepsilon)^{(k-1)}\left(g\left(\boldsymbol{w}^{k},\bm{\Pi}^{k}, \boldsymbol{r}^{k}\right)-g\left(\boldsymbol{w}^*, \boldsymbol{\Pi}^*, \boldsymbol{r}^*\right)\right) \\ \leq D\left(\boldsymbol{r}^* \| \boldsymbol{r}^{1}\right)-(1+\varepsilon)^n D\left(\boldsymbol{r}^* \| \boldsymbol{r}^{n+1}\right) \leq D\left(\boldsymbol{r}^* \| \boldsymbol{r}^{1}\right) \triangleq C,
\end{multline*}
where $C$ is a constant. 

According to Lemma \ref{lem-8}, we obtain
\begin{multline*}
0 \leq \sum_{k=1}^n(1+\varepsilon)^{(k-1)}\left(g\left(\boldsymbol{w}^{n},\bm{\Pi}^{n}, \boldsymbol{r}^{n}\right)-g\left(\boldsymbol{w}^*, \boldsymbol{\Pi}^*, \boldsymbol{r}^*\right)\right) \\
\leq \sum_{k=1}^n(1+\varepsilon)^{(k-1)}\left(g\left(\boldsymbol{w}^{k},\bm{\Pi}^{k}, \boldsymbol{r}^{k}\right)-g\left(\boldsymbol{w}^*, \boldsymbol{\Pi}^*, \boldsymbol{r}^*\right)\right) \leq C
\end{multline*}
when $n$ is sufficient large we have, 
$$
g\left(\boldsymbol{w}^{n},\bm{\Pi}^{n}, \boldsymbol{r}^{n}\right)-g\left(\boldsymbol{w}^*, \boldsymbol{\Pi}^*, \boldsymbol{r}^*\right) \leq \frac{\varepsilon C}{(1+\varepsilon)^n-1}.
$$
When $\varepsilon \geq 0$, we notice that $\varepsilon^{*} = \arg \max_{\varepsilon} \frac{\varepsilon C}{(1+\varepsilon)^n - 1}$ satisfies $(1+\varepsilon^{*})^n = n\varepsilon^{*}(1+\varepsilon^{*})^{n-1}$, i.e., $\max \frac{\varepsilon C}{(1+\varepsilon)^n - 1}  = \frac{C}{n(1+\varepsilon^{*})^{n-1}} \leq \frac{C}{n}$, thus its convergence rate is $O(1/n)$.
\end{proof}

\section{Numerical Experiment} \label{sec:experiment}

In this section, we numerically study the validity of the \alg algorithm. 

We evaluate the RDP functions under two distinct scenarios characterized by different perception measures. The first scenario considers a binary source with Hamming distortion and TV distance perception, while the second involves a Gaussian source with squared-error distortion and KL divergence/Wasserstein-2 perception \cite{wagner2022rate, niu2025rate, xie2025constrained}. For the binary case, since the source is discrete, we can directly specify its probability distribution \( \bm{p} \). In contrast, for the continuous Gaussian source, we apply a discretization procedure to enable numerical computation. Specifically, we first truncate the Gaussian distribution to a finite interval \([ -S, S ]\), and then partition this interval into \(N\) uniformly spaced grid points \(\{ x_i \}_{i=1}^{N}\) with spacing \(\delta = \frac{2S}{N-1}\).
The discrete approximation of the Gaussian distribution is then defined by the probability mass function:
\begin{equation}
    p_i = F\left(x_i + \frac{\delta}{2} \right) - F\left(x_i - \frac{\delta}{2} \right), \quad i = 1, \dots, N,
\end{equation}
where \(F(x)\) denotes the cumulative distribution function of the Gaussian source.

Unless stated otherwise, we assume a Bernoulli parameter \(p = 0.1\) and \(P=0.02\) for the binary source. For the Gaussian case, the default parameters are \(P=0.2\), \(S = 8\), \(\delta = 0.5\), mean \(\mu = 0\), and standard deviation \(\sigma = 2\). In our \alg method, we fix \(\varepsilon = 0.01\).

As shown in Table~\ref{table: KL}, both algorithms achieve comparable accuracy across various parameter settings, while the \alg algorithm demonstrates markedly higher computational efficiency than its predecessor.
This improvement can be attributed to the reduced number of constraints in our modeling process, which, in turn, decreases the number of dual variables to be solved, and thereby reduces the overall problem size.
Such structural simplification also constitutes a significant advantage over the WBM-RDP model.

\begin{table}[H]
\caption{Comparison between \alg algorithm and IAS algorithm under KL divergence}
\label{table: KL}
\centering
\begin{tabular}{c|c|c|c|c|c}
\hline
\multirow{2}{*}{} & \multirow{2}{*}{D} & \multicolumn{2}{c|}{Computation Time (s)} & \multirow{2}{*}{Speed-up ratio} & \multirow{2}{*}{$\|R_\text{AS}-R_\text{Pro}\|_2$} \\ \cline{3-4}
 & & IAS & \multicolumn{1}{|c|}{\alg} &  &  \\ \hline
\multirow{5}{*}{Gaussian source} 
& 1.00 & 0.1853 & $\bm{0.1005}$ & 1.84 & 4.5587E-14 \\
& 2.00 & 0.1930 & $\bm{0.0914}$ & 2.11 & 9.8765E-14 \\
& 3.00 & 0.1739 & $\bm{0.1251}$ & 1.39 & 1.2345E-13 \\
& 4.00 & 0.1480 & $\bm{0.0802}$ & 1.85 & 7.8901E-14 \\
& 5.00 & 0.1687 & $\bm{0.0968}$ & 1.74 & 3.4567E-14 \\ \hline
\end{tabular}
\end{table}

In the second experiment, we perform a comparative analysis between the newly \alg algorithm and the original IAS algorithm under two configurations: binary source and Gaussian source. As demonstrated in Table \ref{table: Wasserstein}, We observe results similar to those in Table \ref{table: KL}.

\begin{table}[H]
\caption{Comparison between \alg algorithm and IAS algorithm under Wasserstein metric}
\label{table: Wasserstein}
\centering
\begin{tabular}{c|c|c|c|c|c}
\hline
\multirow{2}{*}{} & \multirow{2}{*}{D} & \multicolumn{2}{c|}{Computation Time (s)} & \multirow{2}{*}{Speed-up ratio} & \multirow{2}{*}{$\|R_\text{AS}-R_\text{Pro}\|_2$} \\ \cline{3-4}
 & & IAS & \multicolumn{1}{|c|}{\alg} &  &  \\ \hline
\multirow{5}{*}{Binary source}   
& 0.03 & 0.0215 & $\bm{0.0016}$ & 13.44 & 2.9332E-13 \\
& 0.06 & 0.0196 & $\bm{0.0045}$ & 4.36  & 4.5464E-14 \\
& 0.09 & 0.0969 & $\bm{0.0255}$ & 3.80  & 6.6842E-14 \\
& 0.12 & 0.1799 & $\bm{0.0568}$ & 3.17  & 5.9669E-14 \\
& 0.15 & 0.0139 & $\bm{0.0033}$ & 4.21  & 6.1185E-14 \\ \hline
\multirow{5}{*}{Gaussian source} 
& 1.00 & 0.5650 & $\bm{0.3460}$ & 1.63 & 6.2617E-14 \\
& 2.00 & 0.6031 & $\bm{0.3682}$ & 1.64 & 8.2656E-14 \\
& 3.00 & 0.4170 & $\bm{0.3506}$ & 1.19 & 1.0242E-12 \\
& 4.00 & 1.2229 & $\bm{0.4342}$ & 2.82 & 7.0880E-13 \\
& 5.00 & 0.5296 & $\bm{0.2426}$ & 2.18 & 1.6175E-13 \\ \hline
\end{tabular}
\end{table}

\section{Conclusion} \label{sec:conclusion}
In this work, we present a fast algorithm with convergence $O(1/n)$ to compute the information rate distortion perception functions based on our previous Improved-AS algorithm. By adapting the equational relationship of constraints in the original WBM-RDP, we construct the \alg algorithm that is shown to have a $O(1/n)$ convergence rate. Numerical experiments show that the \alg algorithm performs with high accuracy and efficiency.



%
%

\bibliographystyle{spmpsci}      
\bibliography{ref.bib}   


\end{document}